\documentclass[11pt]{article}
\usepackage{fullpage}

\usepackage{dsfont}
\usepackage{amsmath}
\usepackage{amsthm}
\usepackage{amssymb}
\usepackage{graphicx}
\usepackage{url}
\usepackage{nicefrac}
\usepackage{enumitem}
\usepackage{authblk}
\usepackage{tikz} 
\usepackage{subcaption}
\usepackage{xspace}
\usepackage[compact]{titlesec}
\usetikzlibrary{arrows,decorations.markings}

\usepackage{algorithm}
\usepackage[noend]{algorithmic}

\usepackage{typearea}

\usepackage{pgf}
\usetikzlibrary{snakes}
\usetikzlibrary{patterns}
\usetikzlibrary{backgrounds}
\usetikzlibrary{scopes}
\usetikzlibrary{arrows, automata, positioning}
\usetikzlibrary{decorations.markings}
\usepackage[T1]{fontenc}
\usepackage{calc}
\usepackage{color}
\usepackage{lipsum}

\makeatother
\usepackage{color}

\definecolor{blueblack}{rgb}{0,0,.7}

\newcounter{sideremark}

\definecolor{Darkblue}{rgb}{0,0,0.4}
\definecolor{Brown}{cmyk}{0,0.61,1.,0.60}
\definecolor{Purple}{cmyk}{0.45,0.86,0,0}
\definecolor{brickred}{rgb}{0.8, 0.25, 0.33}
\usepackage[colorlinks,citecolor=blue,linkcolor=red,bookmarks=true]{hyperref}

\theoremstyle{plain}
\newtheorem{theorem}{Theorem}
\newtheorem{lemma}{Lemma}
\newtheorem{claim}{Claim}
\newtheorem{definition}{Definition}
\newtheorem{remark}{Remark}
\newtheorem{corollary}{Corollary}

\newcommand{\eps}{\varepsilon}

\newcommand{\calC}{\mathcal{C}}

\newcommand{\hH}{\hat{H}}

\newcommand{\CAP}{\mathrm{CAP}}
\newcommand{\TAP}{\mathrm{TAP}}
\newcommand{\CacAP}{\mathrm{CacAP}}

\usepackage[colorinlistoftodos,textsize=tiny,textwidth=2cm,color=red!25!white,obeyFinal]{todonotes}

\ifdefined\DEBUG
 \newcommand{\fab}[1]{\textcolor{red}{#1}}
  \newcommand{\ja}[1]{\textcolor{blue}{#1}}

  \def\rem#1{{\marginpar{\raggedright\scriptsize #1}}}
  \newcommand{\fabr}[1]{\rem{\textcolor{red}{$\bullet$ #1}}}

\else
  \newcommand{\fab}[1]{#1}
  \newcommand{\fabr}[1]{}
  \newcommand{\ja}[1]{#1}

\fi

\title{Breaching the $2$-Approximation Barrier\\ for Connectivity Augmentation:\\a Reduction to Steiner Tree\thanks{The first author is supported by the NCN grant number 2015/18/E/ST6/00456.
The last 2 authors are partially supported by  the SNSF Excellence Grant 200020B\_182865/1.}}

\author[1]{Jaros\l aw Byrka}
\author[2]{Fabrizio Grandoni}
\author[2]{Afrouz Jabal Ameli}
\affil[1]{University of Wroc\l aw}
\affil[2]{IDSIA, Lugano}

\date{}
\begin{document}
\maketitle
\begin{abstract}
\noindent The basic goal of survivable network design is to build a cheap network that maintains the connectivity between given sets of nodes despite the failure of a few edges/nodes. The \emph{Connectivity Augmentation} Problem ($\CAP$) is arguably one of the most basic problems in this area: given a $k$(-edge)-connected graph $G$ and a set of extra edges (\emph{links}), select a minimum cardinality subset $A$ of links such that adding $A$ to $G$ increases its edge connectivity to $k+1$. Intuitively, one wants to make an \emph{existing} network more reliable by \emph{augmenting} it with extra edges. The best known approximation factor for this NP-hard problem is $2$, and this can be achieved with multiple approaches (the first such result is in [Frederickson and J{\'a}j{\'a}'81]).

It is known [Dinitz et al.'76] that $\CAP$ can be reduced to the case $k=1$, a.k.a. the \emph{Tree Augmentation} Problem ($\TAP$), for odd $k$, and to the case $k=2$, a.k.a. the \emph{Cactus Augmentation} Problem ($\CacAP$), for even $k$. Several better than $2$ approximation algorithms are known for $\TAP$, culminating with a recent $1.458$ approximation [Grandoni et al.'18]. However, for $\CacAP$ the best known approximation is $2$. 

In this paper we breach the $2$ approximation barrier for $\CacAP$, hence for $\CAP$, by presenting a polynomial-time $2\ln(4)-\frac{967}{1120}+\eps<1.91$ approximation. From a technical point of view, our approach deviates quite substantially from the current related literature. In particular, the better-than-2 approximation algorithms for $\TAP$ either exploit greedy-style algorithms  or are based on rounding carefully-designed LPs. These approaches exploit properties of $\TAP$ that do not seem to generalize to $\CacAP$. We instead use a reduction to the Steiner tree problem which was previously used in parameterized algorithms [Basavaraju et al.'14]. This reduction is not approximation preserving, and using the current best approximation factor for Steiner tree [Byrka et al.'13] as a black-box would not be good enough to improve on $2$. To achieve the latter goal, we ``open the box'' and exploit the specific properties of the instances of Steiner tree arising from $\CacAP$. 

In our opinion this connection between approximation algorithms for survivable network design and Steiner-type problems is interesting, and it might lead to other results in the area.  

\end{abstract}

\newpage
\pagenumbering{arabic}

\setcounter{page}{1}

\section{Introduction}

The basic goal of \emph{Survivable Network Design} is to construct cheap networks that provide connectivity guarantees between pre-specified sets of nodes even after the failure of a few edges/nodes (in the following we will focus on the edge failure case). This has many applications, e.g., in transportation and telecommunication networks.


The \emph{Connectivity Augmentation Problem} ($\CAP$) is among the most basic survivable network design problems. Here we are given a $k$-(edge)-connected\footnote{We recall that $G=(V,E)$ is $k$-connected if for every subset of edges $F\subseteq E$, $|F|\leq k-1$, the graph $G'=(V,E\setminus F)$ is connected.} undirected graph $G=(V,E)$ and a collection $L$ of extra edges (\emph{links}). The goal is to find a minimum \emph{cardinality} subset $OPT\subseteq L$ such that $G'=(V,E\cup OPT)$ is $(k+1)$-connected. Intuitively, we wish to augment an existing network to make it more resilient to edge failures. Dinitz et al.~\cite{DKL76} (see also \cite{CJR99,KT93}) presented an approximation-preserving reduction from this problem to the case $k=1$ for odd $k$, and $k=2$ for even $k$. This motivates a deeper understanding of the latter two special cases. 

The case $k=1$ is also known as the \emph{Tree Augmentation Problem} ($\TAP$). The reason for this name is that any $2$-connected component of the input graph $G$ can be contracted, hence leading to a tree. For this problem several better than $2$ approximation algorithms are known \cite{A17,CG18,CG18a,EFKN09,FGKS18,GKZ18,KN16,KN16b,N03}. In particular, the current best approximation factor is $1.458$ \cite{GKZ18}. 

The case $k=2$ is also known as the \emph{Cactus Augmentation Problem} ($\CacAP$), where for similar reasons we can assume that the input graph is a cactus\footnote{We recall that a \emph{cactus} $G$ is a connected undirected graph in which every edge belongs to exactly one cycle. For technical reasons it is convenient to allow length-$2$ cycles consisting of $2$ parallel edges.}. Here the best-known approximation factor is still $2$, and this factor can be achieved with multiple approaches \cite{FJ81,GGPSTW94,J01,KT93}. A better approximation was achieved very recently for the special case where the input cactus is a cycle \cite{GGJS19}.

Hence $2$ is also the best known approximation factor for $\CAP$ in general. One might also observe that $\TAP$ can be easily reduced to $\CacAP$ by duplicating the edges of the input instance. Hence $\CacAP$ and $\CAP$ are equivalent problems in terms of approximability.

\subsection{Our Results and Techniques}

The main result of this paper is the first better than $2$ approximation algorithm for $\CacAP$, hence for $\CAP$. 
\begin{theorem}\label{thr:main}
For any constant $\eps>0$, there is a polynomial-time $2\ln(4)-\frac{967}{1120}+\eps<1.9092+\eps$ approximation algorithm for the Cactus Augmentation problem.
\end{theorem}
\begin{corollary}\label{cor:main}
For any constant $\eps>0$, there is a polynomial-time $2\ln(4)-\frac{967}{1120}+\eps<1.9092+\eps$ approximation algorithm for the Connectivity Augmentation problem. 
\end{corollary}
\begin{proof}
It follows directly from Theorem \ref{thr:main} and the reduction to $\CacAP$ implied by \cite{DKL76}.
\end{proof}

Our result is based on a reduction to the (cardinality) Steiner tree problem by Basavaraju et al. \cite{BFGMRS14}. The authors use this connection to design improved parameterized algorithms (see also \cite{MV15} for a related result).
Recall that in the \emph{Steiner tree} problem we are given an undirected graph $G_{ST}=(T\cup S,E_{ST})$, where $T$ is a set of $t$ \emph{terminals} and $S$ a set of \emph{Steiner nodes}. Our goal is to find a tree (\emph{Steiner tree}) $OPT_{ST}=(T\cup A,F)$ that contains all the terminals (and possibly a subset of Steiner nodes $A$) and has the minimum possible number of edges $|OPT_{ST}|$. Basavaraju et al. observed that, given a $\CacAP$ instance $(G=(V,E),L)$, it is possible to construct (in polynomial time) an \emph{equivalent} Steiner tree instance $G_{ST}=(T\cup L,E_{ST})$. Here $T$ corresponds to the nodes of degree $2$ in $G$, $L$ are the Steiner nodes, and the edges $E_{ST}$ are defined properly (more details in Section \ref{sec:reductionSteiner}). In particular, an optimal solution to $G_{ST}$ induces an optimal solution to $(G,L)$ and vice versa. An example of the reduction is given in Figure \ref{fig:reduction}. Unfortunately, this reduction is not approximation-preserving. In particular, by working out the simple details (see also Section \ref{sec:reductionSteiner}), one obtains that a $\rho_{ST}$-approximation for Steiner tree implies a $\rho\leq 3\rho_{ST}-2$ approximation for $\CacAP$. The current best value of $\rho_{ST}$ is $\ln 4+\eps<1.39$ due to Byrka, Grandoni, Rothvoss and Sanit{\`a} \cite{BGRS13}. Hence this is not good enough to obtain $\rho<2$\footnote{One would need $\rho_{ST}<4/3$ here. Notice that this is not ruled out by the current lower bounds on the approximability of Steiner tree.}. 

In order to obtain our main result we use the same algorithm as in \cite{BGRS13}, but we analyze it differently. In particular, we exploit the specific structure of the instances of Steiner tree arising from $\CacAP$ instances via the above reduction to get a substantially better approximation factor.

In more detail (see also Section \ref{sec:SteinerAlgo}), in the analysis of the algorithm in \cite{BGRS13} one considers an optimal Steiner tree solution $OPT_{ST}=(T\cup A,F)$ rooted at some arbitrary node $r$, marks a random subset $F_{mar}\subseteq F$ of edges so that each Steiner node is connected to some terminal via marked edges, and based on $F_{mar}$ defines a proper (random) \emph{witness set} $W(e)$ for each $e\in F$. The cost of the approximate solution turns out to be at most $(1+\eps)\sum_{e\in F}E[H_{|W(e)|}]$, where $H_i:=1+\frac{1}{2}+\ldots+\frac{1}{i}$ is the $i$-th harmonic number. In particular, the authors show that $E[H_{|W(e)|}]\leq \ln 4$ for each $e\in F$, hence the claimed approximation factor. 

Our analysis of the algorithm deviates from \cite{BGRS13} for the following critical reasons:
\begin{enumerate}\itemsep0pt
\item They (i.e., the authors of \cite{BGRS13}) can assume that each internal node has degree exactly $2$. This can be enforced by exploiting edge weights. We critically need that $OPT_{ST}$ is unweighted, hence we need to deal with arbitrary degrees (which makes the analysis technically more complex).  
\item They mark one child edge of each Steiner node chosen uniformly at random. In our case it is convenient to \emph{favor} child edges with one terminal endpoint (if any). The fact that this helps is not obvious in our opinion.
\item As mentioned above, they provide a per-edge upper bound on $E[H_{|W(e)|}]$. We rather need to average over multiple edges in order to achieve a good bound. Finding a good way to do that is not trivial in our opinion. 
\end{enumerate}    

We remark that, from a technical point of view, our result deviates quite substantially from prior approximation algorithms for $\TAP$. The first improvements on a $2$ approximation where achieved via greedy-style algorithms and a complex case analysis \cite{EFKN09,KN16,KN16b,N03}. More recent approaches are based on rounding stronger and stronger LP (or SDP) relaxations for the problem \cite{A17,CG18,CG18a,FGKS18,GKZ18}. We also use an LP-based rounding algorithm, which is however defined for a generic Steiner tree instance (while the properties of $\TAP$ are used only in the analysis). In our opinion the connection that we established between the approximability of survivable network design problems and Steiner-type problems might lead to other results in the future.

\subsection{Related Work}

One can consider a natural weighted version of $\CAP$ where each link has a positive weight and the goal is to minimize the total weight of selected links. However, in this case the best-know approximation is $2$ even for $\TAP$, and improving on this is a major open problem in the area. The techniques used in this paper seem not to generalize to the weighted case. In particular, one might use a reduction to a node-weighted version of the Steiner tree problem, however the latter problem is harder and in general allows only a logarithmic approximation~\cite{klein1995nearly}. Some progress on weighted $\TAP$ was made in the case of small integer weights. In particular, when the largest weight $W$ is upper bounded by a constant, better than $2$ approximation algorithms are given in \cite{A17,FGKS18,GKZ18}. A technique in \cite{N17} allows one to extend these results to $W=O(\log n)$. Weighted $\TAP$ also admits a $1+\ln 2$ approximation for arbitrary weights if the input tree has constant radius \cite{CN13}. 

A problem closely related to $\CAP$ is to build a minimum size $k$-edge-connected spanning subgraph of a given input graph \cite{CT00,GG12,GR07,JRV03}.

\section{Steiner Tree and Connectivity Augmentation}
\label{sec:Steiner}

In this section we present the mentioned reduction in \cite{BFGMRS14} from $\CacAP$ to Steiner tree (Section \ref{sec:reductionSteiner}). Furthermore, we describe a specific Steiner tree approximation algorithm that we will use to solve the instance arising from the above reduction (Section \ref{sec:SteinerAlgo}). We analyze the resulting approximation factor in Section \ref{sec:apx}.

\subsection{A Reduction to Steiner Tree}
\label{sec:reductionSteiner}


Consider a $\CacAP$ instance $(G=(V,E),L)$. For a link $\ell=(v_0,v_{q+1})$, let $v_1,\ldots,v_q$ be the sequence of nodes of degree at least $4$ other than $v_0$ and $v_{q+1}$ that lie along every simple $v_0$-$v_{q+1}$ path. Notice that each pair $\ell_i=\{v_i,v_{i+1}\}$ \
lies along a distinct cycle $C_i$ visited by the mentioned path. We call each such $\ell_{i}$ the \emph{projection} of $\ell$ on $C_i$. Consider two links $\ell=\{x,y\}$ and $\ell'=\{x',y'\}$ that have endpoints in the same cycle $C$. Then we say that $\ell$ and $\ell'$ \emph{cross} if one of the following two conditions hold: (1) they share one endpoint or (2) taking one simple $x$-$y$ path $P$ along $C$, $P$ contains exactly one node in $\{x',y'\}$ as an internal node. We say that any two links $\ell$ and $\ell'$ \emph{cross} if there exists a projection $\ell_i$ of $\ell$ and a projection $\ell'_j$ of $\ell'$ such that $\ell_i$ and $\ell'_j$ cross. See Figure \ref{fig:reduction} (left) for an example.

From $(G,L)$ we construct a Steiner tree instance $G_{ST}=(T\cup S,E_{ST})$ as follows. For each one of the $t$ nodes $v$ of degree 2 in $G$, add a terminal $v$ to $T$; for each link $\ell\in L$, add a Steiner node $\ell$ to $S$ (i.e., $S=L$); for each $\ell\in L$ and endpoint $v\in T$ of $\ell$, add $\{\ell,v\}$ to $E_{ST}$; finally, for any two links $\ell$ and $\ell'$ that \emph{cross}, add $\{\ell,\ell'\}$ to $E_{ST}$. See Figure \ref{fig:reduction} (right) for an example. We observe the following simple facts.
\begin{remark}\label{rem:adjacentTerminals}
Each Steiner node is adjacent to at most $2$ terminals.
\end{remark}
\begin{remark}\label{rem:clique}
The neighbors of each terminal are Steiner nodes and form a clique.
\end{remark}
We will critically exploit the following lemma sketched in \cite{BFGMRS14} (Lemma 1). For the sake of completeness we give a (more detailed) proof of it in Appendix \ref{sec:lemma1proof}.
\begin{lemma}\cite{BFGMRS14}\label{lem:SteinerReduction}
$A\subseteq L$ is a feasible solution to a $\CacAP$ instance $(G,L)$ iff, in the corresponding Steiner tree instance $G_{ST}=(T\cup L,E_{ST})$,  $G_{ST}[T\cup A]$ is connected. 
\end{lemma}


Notice that the above reduction is not approximation-preserving. Still, we can state the following. 
\begin{corollary}\label{cor:SteinerReduction}
The optimum solution $OPT$ to the input $\CacAP$ instance, induces a solution $OPT_{ST}$ of cost $|OPT_{ST}|=|OPT|+t-1$ for the associated Steiner tree instance.
Vice versa, given a solution $APX_{ST}$ to the Steiner tree instance, one can construct in polynomial time a solution $APX$ to the input $\CacAP$ instance with $|APX|=|APX_{ST}|-t+1$.
\end{corollary}
\begin{proof}
Both claims follow directly from Lemma \ref{lem:SteinerReduction}. For the first claim, it is sufficient to observe that a spanning tree of $G_{ST}[T\cup OPT]$ contains $t+|OPT|-1$ edges. For the second claim, observe that the Steiner nodes in $APX_{ST}$ induce a feasible solution to $\CacAP$. The claim follows since $|APX_{ST}|=s+t-1$, where $s$ is the number of Steiner nodes in $APX_{ST}$.
\end{proof}
We will exploit also the following simple fact.
\begin{lemma}\label{lem:deg1Terminals}
There is a feasible solution $OPT_{ST}$ to the Steiner tree instance with $|OPT_{ST}|=|OPT|+t-1$ where terminals have degree exactly $1$.
\end{lemma}
\begin{proof}
Given any feasible solution $ST$ to the problem, we can transform it into a solution $ST'$ of the same cost where some terminal $v$ of degree $d(v)\geq 2$ in $ST$ has degree $d(v)-1$ in $ST'$. In order to do that, consider any terminal $v$ adjacent to two Steiner nodes $\ell$ and $\ell'$ in $ST$. By Remark \ref{rem:clique}, $\ell$ and $\ell'$ are adjacent. Hence $ST':=ST\cup \{\ell,\ell'\}\setminus \{v,\ell'\}$ is a feasible Steiner tree of the same cost and with the desired property.   

By iteratively applying the above process to the solution $OPT_{ST}$ guaranteed by Corollary \ref{cor:SteinerReduction} one obtains the desired solution. 
\end{proof}
As mentioned earlier, a $\rho_{ST}$ approximation for Steiner tree (used as a \emph{black box}) provides a $3\rho_{ST}-2$ approximation for $\CacAP$ by the above construction. Indeed, the Steiner tree instance has cost at most $|OPT|+t-1$ by Corollary \ref{cor:SteinerReduction}, hence an approximate solution $APX_{ST}$ would cost at most $\rho_{ST}(|OPT|+t-1)$. By the same corollary, we can convert this into a solution $APX$ to $\CacAP$ of cost at most $\rho_{ST}(|OPT|+t-1)-t+1$. Next observe that $|OPT|\geq t/2$. Indeed, any node of degree $2$ in the $\CacAP$ instance needs to have at least one link incident to it in a feasible solution, and a link can be incident to at most $2$ such nodes. Thus $|APX|\leq 3\rho_{ST}|OPT|-2|OPT|$. In order to improve on this simple bound, we will have to \emph{open the box}.

\subsection{Steiner Tree via Iterative Randomized Rounding}
\label{sec:SteinerAlgo}

As we mentioned in the introduction, the current best $(\ln 4+\eps)$-approximate Steiner tree algorithm from \cite{BGRS13}, used as a black box, is not good enough to break the $2$-approximation barrier for $\CacAP$. However, it turns out that the same algorithm achieves this goal in combination with a different analysis that exploits the properties of the specific Steiner tree instances arising from $\CacAP$. 

We next sketch the basic properties of the algorithm and analysis in \cite{BGRS13} that we need here. A more detailed description is given in Section \ref{sec:backgroundSteiner} in the Appendix for the sake of completeness.
The authors of \cite{BGRS13} consider an LP relaxation $DCR_k$ for the problem based on \emph{directed $k$-components} for a proper constant parameter $k$ depending on $\eps$. They iteratively solve this LP, sample a directed $k$-component $C$ with probability proportional to the LP values, and contract $C$. The process ends when all terminals are contracted into one node. This algorithm can be derandomized, and the deterministic version is good enough for our application. We do not need more details about this algorithm, other than that it runs in polynomial time.  

In the analysis the authors of \cite{BGRS13} consider any feasible Steiner tree $ST=(T\cup A,F)$, which is seen as rooted at some arbitrary node $r$. 
Then the authors define a \emph{marking scheme} where some child edge of each internal (Steiner) node is marked. 
A given marking scheme defines a \emph{witness set} $W(e)$ for each edge $e$: this consists of pairs of terminals $\{t',t''\}$ such that the (simple) $t'$-$t''$ path in $ST$ contains $e$ and precisely one unmarked edge.
We let $w(e)=|W(e)|$. Notice that $w(e)=1$ for an unmarked edge. Then the authors prove the following, where $H_i:=1+\frac{1}{2}+\ldots+\frac{1}{i}$ is the $i$-th harmonic number.   
\begin{lemma}\label{lem:SteinerApx}\cite{BGRS13}
For any feasible Steiner tree $ST=(T\cup A,F)$ and marking scheme, for a large enough parameter $k=O_\eps(1)$, the cost of the solution computed by the above algorithm is at most 
$
(1+\eps)\sum_{e\in F}E[H_{w(e)}].
$
\end{lemma}


\section{An Improved $\CacAP$ Approximation Algorithm}
\label{sec:apx}

In this section we present our improved approximation for $\CacAP$. The algorithm is rather simple: we just build the Steiner tree instance $G_{ST}=(T\cup L,E_{ST})$ associated with the input $\CacAP$ instance $(G,L)$ and compute an approximate solution $APX_{ST}$ to $G$ via the algorithm in \cite{BGRS13} sketched in Section \ref{sec:SteinerAlgo}. Then we derive from $APX_{ST}$ a feasible solution $APX$ to the input $\CacAP$ instance as described in Corollary \ref{cor:SteinerReduction}. We let $apx$ denote the approximation ratio of this algorithm.

In Section \ref{sec:marking} we describe our alternative marking scheme and prove some of its properties. In Section \ref{sec:analysis} we complete the analysis of the approximation factor.

\subsection{An Alternative Marking Scheme}
\label{sec:marking}

Recall that in the analysis of the Steiner tree approximation algorithm in \cite{BGRS13}, one can focus on a specific feasible Steiner tree $ST$ and on a specific marking scheme (so that Steiner nodes are connected to some terminal via paths of marked edges). As feasible solution $ST$ we consider the solution $OPT_{ST}=(T\cup OPT,F)$, of cost $|OPT|+t-1$ and with terminals being leaves, guaranteed by Lemma \ref{lem:deg1Terminals}. 

We mark edges in the following way. Let us root $OPT_{ST}$ at some Steiner node $r$ which is adjacent to at least one terminal. For a Steiner node $\ell$, we let $d(\ell)$, $s(\ell)$ and $t(\ell)$ be the number of its children, Steiner children, and terminal children, resp. In particular $d(\ell)=s(\ell)+t(\ell)$ and (by Remark \ref{rem:adjacentTerminals}) $t(\ell)\leq 2$.

For each link node $\ell$, there are two options. If $\ell$ has at least one terminal child, we select one such child $t$ uniformly at random, and mark edge $\{\ell,t\}$. Otherwise, we choose a child $\ell'$ of $\ell$ ($\ell'$ being a Steiner node) uniformly at random, and mark edge $\{\ell,\ell'\}$. Notice that this is obviously a feasible marking scheme. Observe also that in \fab{our} marking we \emph{favor} edges connecting Steiner nodes to terminals: this will be critical in our analysis. See Figure \ref{fig:optimalSteiner} for a possible marking of this type.

Let $APX_{ST}$ be the Steiner tree computed by the algorithm.
Let $F_{mar}$ and $F_{unm}$ be the (random) sets of marked and unmarked edges, resp., that partition $F$. Recall that for each $e\in F$, there exists a (random) witness set $W(e)$ of size $w(e)=|W(e)|$. Observe that each Steiner node $\ell$ has precisely one marked child edge $m(\ell)$. We let the \emph{cost} $c(\ell)$ of $\ell$ be $E[H_{w(m(\ell))}]$. The following bound on the approximation ratio holds.
\begin{lemma}\label{lem:apx1}
$
apx\leq 2\eps+\frac{1+\eps}{|OPT|} \sum_{\ell\in OPT}c(\ell).
$
\end{lemma}
\begin{proof}
Recall that by Lemma \ref{lem:SteinerApx} the expected cost of the computed Steiner tree $APX_{ST}$ is, modulo a factor $(1+\eps)$, at most
\begin{eqnarray*}
& E[\sum_{e\in F}H_{w(e)}]=E[\sum_{e\in F_{mar}}H_{w(e)}+\sum_{e\in F_{unm}}H_{w(e)}]\\
= & E[\sum_{e\in F_{mar}}H_{w(e)}+|F_{unm}|]=E[\sum_{e\in F_{mar}}H_{w(e)}]+t-1.
\end{eqnarray*}
In the second-last equality above we used the fact that $w(e)=1$ deterministically for an unmarked edge, and in the last equality above the fact that there are precisely $|OPT|$ marked edges and consequently exactly $t-1$ unmarked ones. From $APX_{ST}$ we derive a feasible solution $APX$ to the input instance of cost $|APX|=|APX_{ST}|-1+t$ by Corollary \ref{cor:SteinerReduction}. Hence 
$$
|APX|\leq (1+\eps)(E[\sum_{e\in F_{mar}}H_{w(e)}]+t-1)-1+t\leq (1+\eps)E[\sum_{e\in F_{mar}}H_{w(e)}]+2\eps |OPT|. 
$$
In the last inequality above we used the trivial lower bound $|OPT|\geq t/2$ that we mentioned earlier. The claim follows since by definition $\sum_{e\in F_{mar}}E[H_{w(e)}]=
\sum_{\ell\in OPT}c(\ell)$.
\end{proof}
From the above lemma, modulo factors $(1+\eps)$, the approximation ratio of our algorithm is given by the \emph{average cost} of  Steiner nodes. The following lemma gives a generic upper bound on the cost for each non-root Steiner node based on the degree sequence of its ancestors\footnote{Observe that for the root $r$, $c(r)=H_{d(r)-1}$ deterministically.}. 
\begin{lemma}\label{lem:costPerNode}
Given a non-root Steiner node $\ell$, let $\ell_q$ be the lowest proper ancestor\footnote{Observe that this ancestor exists since the root has this property by assumption.} of $\ell$ with $t(\ell_q)>0$. Let $\ell=\ell_1,\ell_2,\ldots,\ell_q$, $q\geq 2$, be the simple path between $\ell$ and $\ell_q$, and let $d_i=d(\ell_i)$. Then
$$
c(\ell)=\sum_{h=1}^{q-2}\frac{(d_{h+1}-1)H_{d_1+\ldots+d_h-h+1}}{d_2\cdot \ldots \cdot d_{h+1}}+\frac{H_{d_1+\ldots+d_{q-1}-q+2}}{d_2\cdot \ldots \cdot d_{q-1}}.\label{eqn:c(v)}
$$
\end{lemma}
\begin{proof}
By definition $c(\ell)=c(\ell_1)=E[H_{w(e)}]$, where $e=m(\ell_1)=\{\ell_1,\ell_0\}$ is the marked child edge of $\ell_1$. Recall that $W(e)$ contains one entry for each path in the tree that contains $e$ and precisely one unmarked edge.  In our specific case, condition on $\{\ell_0,\ell_1\},\{\ell_1,\ell_2\},\ldots,\{\ell_{h-1},\ell_h\}$ being a maximal sequence of consecutive marked edges. Notice that by construction $\{\ell_{q-1},\ell_q\}$ is unmarked (since $\ell_q$ has a terminal child by definition), hence $h\leq q-1$. In this case $w(e)=d_1+\ldots+d_h-(h-1)$. For $h<q-1$, the mentioned event happens with probability $\frac{1}{d_2}\cdot \ldots \cdot  \frac{1}{d_{h}} \cdot \frac{d_{h+1}-1}{d_{h+1}}$. For $h=q-1$, this probability is $\frac{1}{d_{2}}\cdot \ldots \cdot \frac{1}{d_{h}}$. The claim follows by computing the expectation of $H_{w(e)}$.
\end{proof}    

We next provide an upper bound on $c(\ell)$ as a function of $d(\ell)$ only. Let us define the following variant of $H_i$:
$$
\hat{H}_i:=\frac{1}{2}H_i+\frac{1}{4}H_{i+1}+\ldots=\sum_{j\geq 0}\frac{1}{2^{j+1}}H_{i+j}.
$$
One has that $\hat{H}_1=\ln(4)$ and $\hat{H}_{j+1}=2\hat{H}_{j}-H_j$. Notice that, modulo an additive $\eps$, $\hat{H}_1$ is precisely the approximation factor for Steiner tree achieved in \cite{BGRS13}. The first few approximate values of $\hat{H}_i$ are $\hat{H}_1<1.3863$, $\hat{H}_2<1.7726$, $\hat{H}_3<2.0452$, $\hat{H}_4< 2.2571$, $\hat{H}_5<2.4308$, $\hat{H}_6<2.5781$, $\hat{H}_7<2.7062$,  and $\hat{H}_8<2.8195$.

The proof of the following lemma, though not entirely trivial, is mostly based on algebraic manipulations and therefore we postpone it to the appendix.

\begin{lemma}\label{lem:badBound}
For any $\ell\in OPT$, $c(\ell)\leq \hat{H}_{d(\ell)}$.
\end{lemma}
In next subsection we will see that for a carefully defined subset of Steiner nodes $\ell$ it is possible to obtain a better upper bound on $c(\ell)$ than the one provided by Lemma \ref{lem:badBound}. This will be critical in our analysis since the latter bound is not strong enough.

\subsection{Analysis of the Approximation Factor}
\label{sec:analysis}

In this section we upper bound the approximation factor $apx$ as given by  Lemmas \ref{lem:apx1} and \ref{lem:costPerNode}. In order to simplify our analysis, it is convenient to focus our attention on a specific class of well-structured Steiner trees $OPT_{ST}$ (see also Figure \ref{fig:optimalSteiner}). The following lemma shows that this is (essentially) w.l.o.g.
\begin{definition}\label{def:wellStructured}
A rooted Steiner tree is well-structured if, for every Steiner node $\ell$: (1) $\ell$ has at least $2$ children and (2) $\ell$ has $0$ or $2$ terminal children. 
\end{definition}
\begin{lemma}\label{lem:wellStructured}
Let $\rho$ be the supremum of $\rho(OPT_{ST})=\frac{1}{|OPT|}\sum_{\ell\in OPT}c(\ell)$ over Steiner trees $OPT_{ST}=(T\cup OPT,F)$, and $\rho_{ws}$ be the same quantity computed over the subset of well-structured Steiner trees  $OPT_{ST}$ of the mentioned type. Then $\rho\leq \max\{\hH_1,\rho_{ws}\}$. 
\end{lemma}
\begin{proof}
Recall that in $OPT_{ST}$ each Steiner node $\ell$ has at most $2$ terminal children. Consider any such tree where some Steiner node $\ell'$ has precisely one terminal child $t$. Consider the tree $OPT'_{ST}$ which is obtained from $OPT_{ST}$ by appending to $\ell'$ a second terminal child $t'$. Observe that the value of $c(\ell)$ does not decrease for any $\ell$, and it increases for $\ell=\ell'$. Thus $\rho(OPT'_{ST})>\rho(OPT_{ST})$. Hence $\rho$ is equal to the supremum of $\rho(OPT_{ST})$ over the subfamily of trees that satisfies (2) in Definition \ref{def:wellStructured}. 

Now consider any tree $OPT_{ST}$ that satisfies (2), and let $o(OPT_{ST})$ be the number of its Steiner nodes with precisely one child. We prove by induction on $o(OPT_{ST})$ that $\rho(OPT_{ST})\leq \max\{\hH_1,\rho_{ws}\}$. The claim is trivially true for $o(OPT_{ST})=0$ since in this case $OPT_{ST}$ is well-structured. Assume the claim is true up to $q-1\geq 0$, and consider $OPT_{ST}=(T\cup OPT,F)$ with $o(OPT_{ST})=q$. Let $\ell'$ be any Steiner node with precisely one child $\ell''$. Observe that $\ell''$ has to be a Steiner node as well by (2), and that $c(\ell')\leq \hH_1$ by Lemma \ref{lem:badBound}. Consider the tree $OPT'_{ST}=(T\cup OPT',F')$ obtained by contracting edge $(\ell',\ell'')$. We observe that $OPT'_{ST}$ satisfies (2), $o(OPT'_{ST})=q-1$ and $|OPT'|=|OPT|-1$. Note also that for any Steiner node $\ell$ different from $\ell'$ and $\ell''$ the value of $c(\ell)$ does not change, while for the new node $\tilde{\ell}$ resulting from the contraction one has $c(\tilde{\ell})=c(\ell'')$. We can conclude that
\begin{eqnarray*}
& \frac{1}{|OPT|}\sum_{\ell\in OPT}c(\ell)\leq \frac{1}{|OPT|}(\hH_1+\sum_{\ell\in OPT\setminus \{\ell''\}}c(\ell))\\
= & \frac{1}{|OPT|}(\hH_1+\sum_{\ell \in OPT'}c(\ell))\leq \max\{\hH_1,\frac{1}{|OPT'|}\sum_{\ell \in OPT'}c(\ell)\}\leq \max\{\hH_1,\rho_{ws}\},
\end{eqnarray*}
where in the last inequality we used the inductive hypothesis.
\end{proof}

We next show an upper bound on $\rho_{ws}$ which is strictly greater than $\hH_1$. It then follows from Lemma \ref{lem:wellStructured} that the same upper bound holds on $\rho$. For this goal, we next assume that $OPT_{ST}$ is well-structured.

The upper bound on $c(\ell)$ from Lemma \ref{lem:badBound} is not sufficient to achieve a good approximation factor. In order to achieve a tighter bound, we consider the following classification of the Steiner nodes (see also Figure \ref{fig:optimalSteiner}). 
\begin{definition}
A Steiner node $\ell'$ is a \emph{good father} if it has at least one terminal child (hence precisely $2$ such children by the above assumptions), and a \emph{bad father} otherwise. Each Steiner child $\ell$ of a good father $\ell'$ is \emph{good}, and all other Steiner nodes are \emph{bad}. Let $OPT_{gf}$, $OPT_{bf}$, $OPT_{g}$ and $OPT_{bad}$ denote the sets of good fathers, bad fathers, good nodes and bad nodes, resp.
\end{definition}
Notice that the above classification is not affected by the random choices in the marking scheme. For good nodes, the analysis of the cost can be refined as follows.
\begin{lemma}\label{lem:goodBound}
For any $\ell\in OPT_{g}$, $c(\ell)\leq H_{d(\ell)}$.
\end{lemma}   
\begin{proof}
Suppose $\ell$ has a parent $\ell'$, which is a good father by definition. This implies that the edge $(\ell',\ell)$ is deterministically unmarked, hence $w(m(\ell))=d(\ell)$ deterministically. If $\ell$ has no parent (i.e.\ja{,} it is the root $r$), then $w(m(\ell))=d(\ell)-1$. The claim follows.
\end{proof}
Putting everything together, we obtain the following.
\begin{lemma}\label{lem:apx2}
$apx\leq 2\eps+\frac{1+\eps}{|OPT|}\sum_{\ell\in OPT}c'(\ell)$ where 
$
c'(\ell)=\begin{cases}H_{d(\ell)} & \text{if }\ell\in OPT_g;\\
\hat{H}_{d(\ell)} & \text{if }\ell\in OPT_b.\end{cases}$ 
\end{lemma}
\begin{proof}
It follows from Lemma \ref{lem:apx1}, by replacing $c(\ell)$ as in Lemma \ref{lem:costPerNode} with the upper bounds given by Lemmas \ref{lem:badBound} and \ref{lem:goodBound}.
\end{proof}


We rewrite the upper bound from Lemma \ref{lem:apx2} as follows. Let $p\in [0,\hat{H}_2-H_2]$ be a parameter to be fixed later. Intuitively, each good Steiner node $\ell\in OPT_{g}$ pays a \emph{present} $p$ to its (good) father $\ell'\in OPT_{gf}$ to thank $\ell'$ for making itself good. This increases the cost of $\ell$ by $p$. Symmetrically, each good father $\ell'\in OPT_{gf}$ collects presents from its (good) Steiner children and uses them to lower its own cost. Clearly by definition the total modification of the cost is zero. Let us call $c''(\ell)$ the modified costs. Then one obtains the following equality:
\begin{eqnarray}
\frac{1}{|OPT|}\sum_{\ell\in OPT}c'(\ell)=\frac{1}{|OPT|}\sum_{\ell\in OPT}c''(\ell)\label{eqn:ub3}
\end{eqnarray}
where
$$
c''(\ell)=
\begin{cases}
H_{d(\ell)}+p-s(\ell)p & \text{if } \ell\in OPT_{g}\cap OPT_{gf};\\
H_{d(\ell)}+p & \text{if } \ell\in OPT_{g}\cap OPT_{bf}; \\
\hat{H}_{d(\ell)}-s(\ell)p & \text{if } \ell\in OPT_{b}\cap OPT_{gf};\\
\hat{H}_{d(\ell)} & \text{if } \ell\in OPT_{b}\cap OPT_{bf}.
\end{cases}
$$

In order to upper bound \eqref{eqn:ub3}, we partition $OPT$ into groups of nodes as follows (see also Figure \ref{fig:optimalSteiner}). 
\begin{definition}
A Steiner node $\ell$ is \emph{leaf-Steiner} if it has no Steiner children (i.e., $d(\ell)=t(\ell)=2$) and \emph{internal-Steiner} otherwise (i.e., $s(\ell)>0$). We let $OPT_{lf}$ and $OPT_{in}$ be the set of leaf-Steiner and internal-Steiner nodes, resp.
\end{definition}
We associate to each $\ell\in OPT_{in}$ a distinct subset $OPT_{lf}(\ell)$ of precisely $s(\ell)-1$ leaf-Steiner nodes, and let $g(\ell)=\{\ell\}\cup OPT_{lf}(\ell)$ be the \emph{group} of $\ell$. The mapping is constructed iteratively in a bottom-up fashion as follows. Initially all Steiner nodes are unprocessed. We maintain the invariant that the subtree rooted at an unprocessed leaf-Steiner node or at a processed node with unprocessed parent contains precisely one unprocessed leaf-Steiner node. Clearly the invariant holds at the beginning of the process. We consider any unprocessed  internal-Steiner node $\ell$ whose Steiner descendants are either processed or leaf-Steiner nodes. By the invariant, each subtree rooted at a Steiner child of $\ell$ (which is either an unprocessed leaf-Steiner node or a processed internal-Steiner node) contains one unprocessed leaf-Steiner node. Among this set of $s(\ell)$ unprocessed leaf-Steiner nodes, we select arbitrarily a set $OPT_{lf}(\ell)$ of size $s(\ell)-1$ and set $g(\ell)=\{\ell\}\cup OPT_{lf}(\ell)$. All nodes in $g(\ell)$ are marked as processed. Observe that the subtree rooted at $\ell$ still contains an unprocessed leaf-Steiner node, hence the invariant is preserved in the following steps.   
At the end of the process (i.e., after processing the root $r$) there will be precisely one leaf-Steiner node $\ell^*$ which is still unprocessed, which  forms a special group $g(\ell^*)=\{\ell^*\}$ on its own. Notice that the groups define a partition of $OPT$. In particular, $OPT=\{\ell^*\}\cup \bigcup_{\ell\in OPT_{in}}g(\ell)$. Notice also that $|g(\ell)|=s(\ell)$ for all $\ell\in OPT_{in}$ (while $|g(\ell^*)|=1$).

Let $a(\ell)$ be the average value of $c''(\cdot)$ over the elements of $g(\ell)$. Then obviously the maximum value of $a(\ell)$ over the groups upper bounds the average value of $c''(\cdot)$:
\begin{equation}
\frac{1}{|OPT|}\sum_{\ell\in OPT}c''(\ell)\leq \max_{\ell\in OPT_{in}\cup \{\ell^*\}}\{a(\ell)\}.\label{eqn:ub4}
\end{equation}
For $\ell=\ell^*$ one has that $a(\ell^*)=c''(\ell^*)=\hat{H}_2$ if $\ell^*$ is bad, and $a(\ell^*)=c''(\ell^*)=H_2+p\leq \hat{H}_2$ otherwise. For the other groups $g(\ell)$, there is always a subset of $s(\ell)-1$ leaves whose contribution to the cost is at most $\hat{H}_2$ each by the same argument as above. Furthermore, we have to add the cost $c''(\ell)$. We can conclude that:
$$
a(\ell)\leq
\begin{cases}
a_1(s(\ell)):=\frac{H_{s(\ell)+2}+p-s(\ell)p+(s(\ell)-1)\hat{H}_2}{s(\ell)} & \text{if } \ell\in OPT_{g}\cap OPT_{gf};\\
a_2(s(\ell)):=\frac{H_{s(\ell)}+p+(s(\ell)-1)\hat{H}_2}{s(\ell)} & \text{if } \ell\in OPT_{g}\cap OPT_{bf}; \\
a_3(s(\ell)):=\frac{\hat{H}_{s(\ell)+2}-s(\ell)p+(s(\ell)-1)\hat{H}_2}{s(\ell)} & \text{if } \ell\in OPT_{b}\cap OPT_{gf};\\
a_4(s(\ell)):=\frac{\hat{H}_{s(\ell)}+(s(\ell)-1)\hat{H}_2}{s(\ell)} & \text{if } \ell\in OPT_{b}\cap OPT_{bf};\\
\hat{H}_2 & \text{if } \ell=\ell^*.
\end{cases}
$$
In the first and third case above we used the fact that $d(\ell)=s(\ell)+2$ ($\ell$ is a good father, hence has $2$ terminal children), while in the second and fourth case the fact that $d(\ell)=s(\ell)$ ($\ell$ is a bad father, hence has no terminal child). 

We are now ready to prove the main result of this paper. 
\begin{proof}[Proof of Theorem \ref{thr:main}]
Consider the above algorithm. Combining Lemma \ref{lem:apx2} with \eqref{eqn:ub3} and \eqref{eqn:ub4} one gets 
\begin{equation}
apx\leq 2\eps+(1+\eps)\max_{i\geq 1}\{\hat{H}_2,a_1(i),a_2(i),a_3(i),a_4(i)\}.\label{eqn:ub5}
\end{equation}
We need the following technical result (proof in Appendix).
\begin{claim}\label{claim:maximum}
For any $p\in [0,\hat{H}_2-H_2]$, the maximum of $a_1(i)$, $a_2(i)$, $a_3(i)$, and $a_4(i)$ is achieved for $i$ at most $6$, $8$, $6$ and $8$, resp. 
\end{claim}

From \eqref{eqn:ub5} and Claim \ref{claim:maximum}, for any $p\in [0,\hat{H}_2-H_2]$, one has
\begin{equation}
apx\leq 2\eps+(1+\eps)\max\{\hat{H}_2,\max_{1\leq i\leq 6}\{a_1(i)\},\max_{1\leq i\leq 8}\{a_2(i)\},\max_{1\leq i\leq 6}\{a_3(i)\},\max_{1\leq i\leq 8}\{a_4(i)\}\}.\label{eqn:ub6}
\end{equation}
Numerically the minimum of the right-hand side of \eqref{eqn:ub6} is achieved for $p\simeq 0.135$, and the two largest values inside the maximum turn out to be $a_2(7)$ and $a_3(1)$.
By imposing $\frac{H_7+6\hat{H}_2+p}{7}=a_2(7)=a_3(1)=\hat{H}_3-p$ one gets $p=\frac{7\hat{H}_3-H_7-6\hat{H}_2}{8}$. For that value of $p$ the value of the maximum is precisely $\frac{H_7+6\hat{H}_2+\hat{H}3}{8}=2\ln 4-\frac{967}{1120}$. The claim follows by scaling $\eps$ properly.
\end{proof}

\newpage

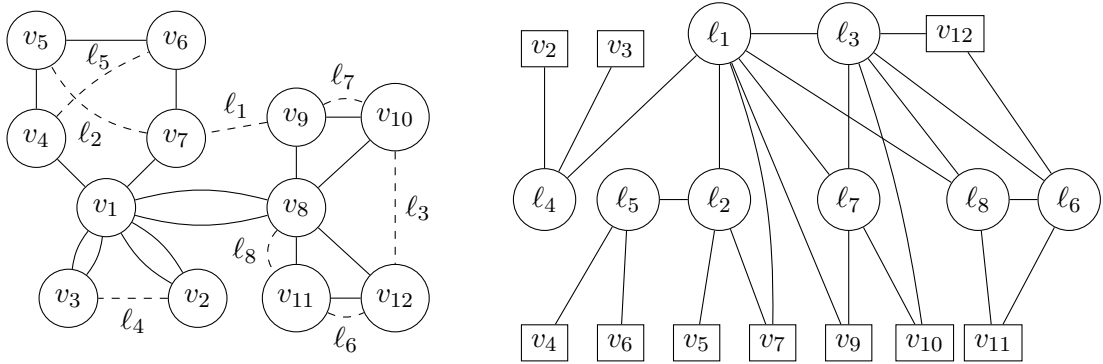
\begin{figure}    
\begin{tikzpicture}
	\tikzstyle{every state}=[fill=white,draw,minimum size=10pt]
\node[state, circle] (1) {$v_1$};
\node[state, circle, below right of=1, node distance=1.7cm] (2) {$v_2$};
\node[state, circle, left of=2, node distance=1.7cm] (3) {$v_3$};
\node[state, circle, right of=1, node distance=2.5cm] (8) {$v_8$};
\node[state, circle, above of=8, node distance=1.2cm] (9) {$v_9$};
\node[state, circle,right of=9, node distance=1.3cm] (10) {$v_{10}$};
\node[state, circle, below of=8, node distance=1.2cm] (11) {$v_{11}$};
\node[state, circle, right of=11, node distance=1.3cm] (12) {$v_{12}$};
\node[state, circle, above right of=1, node distance=1.3cm] (7) {$v_7$};
\node[state, circle, above left of=1, node distance=1.3cm] (4) {$v_4$};
\node[state, circle, above of=7, node distance=1.3cm] (6) {$v_6$};
\node[state, circle, above of=4, node distance=1.3cm] (5) {$v_5$};

\path (1) edge[bend right=15] (8);
\path (1) edge[bend left=15] (8);
\path (1) edge[bend left=15] (2);
\path (1) edge[bend right=15] (2);
\path (1) edge[bend left=15] (3);
\path (1) edge[bend right=15] (3);
\path (8) edge (9);
\path (8) edge (10);
\path (10) edge (9);
\path (8) edge (11);
\path (8) edge (12);
\path (12) edge (11);
\path (1) edge (4);
\path (4) edge (5);
\path (5) edge (6);
\path (6) edge (7);
\path (7) edge (1);

\path [dashed] (9) edge [bend right=0] node[above] {$\ell_1$}(7);
\path [dashed] (7) edge[bend left=25] node[below] {$\ell_2$} (5);
\path [dashed] (4) edge[bend left=12] node[above] {$\ell_5$} (6);
\path[dashed] (2) edge node[below] {$\ell_4$}(3);
\path[dashed] (10) edge node[right]{$\ell_3$}(12);
\path[dashed] (11) edge[bend right=25] node[below] {$\ell_6$} (12);
\path[dashed] (9) edge[bend left=25] node[above] {$\ell_7$} (10);

\path[dashed] (8) edge[bend right=40] node[left] {$\ell_8$} (11);

\end{tikzpicture}
\hspace{20pt}
\begin{tikzpicture}
	\tikzstyle{every state}=[fill=white,draw,minimum size=5pt]

\node[state,  circle] (l1) {$\ell_1$};
\node[state, circle, right of=l1, node distance=1.7cm] (l3) {$\ell_3$};
\node[state, circle, below of=l1, node distance=2.2cm] (l2) {$\ell_2$};
\node[state, circle, left of=l2, node distance=1.2cm] (l5) {$\ell_5$};
\node[state, circle, left of=l5, node distance=1.1cm] (l4) {$\ell_4$};
\node[state, circle, right of=l2, node distance=1.7cm] (l7) {$\ell_7$};
\node[state, circle, right of=l7, node distance=1.7cm] (l8) {$\ell_8$};
\node[state, circle, right of=l8, node distance=1.2cm] (l6) {$\ell_6$};

\node[state, rectangle, above of=l4,node distance= 2cm] (2) {$v_2$};
\node[state, rectangle, right of=2,node distance= 1cm] (3) {$v_3$};
\node[state, rectangle, below of=l4,node distance= 1.9cm] (4) {$v_4$};
\node[state, rectangle, right of=4,node distance= 1cm] (6) {$v_6$};
\node[state, rectangle, right of=6,node distance= 1cm] (5) {$v_5$};
\node[state, rectangle, right of=5,node distance= 1cm] (7) {$v_7$};
\node[state, rectangle, right of=7,node distance= 1cm] (9) {$v_9$};
\node[state, rectangle, right of=9,node distance= 1cm] (10) {$v_{10}$};
\node[state, rectangle, right of=10,node distance= 0.9cm] (11) {$v_{11}$};
\node[state, rectangle, right of=l3,node distance= 1.4cm] (12) {$v_{12}$};

\path (l1) edge[bend left=10] (7); 
\path (l1) edge (9);
\path (l3)edge[bend left=5] (10);
\path (l3)edge (12);
\path (l2) edge (5);
\path (l2) edge (7);
\path (l4) edge (2);
\path (l4) edge (3);
\path (l5) edge (4);
\path (l5) edge (6);
\path (l7) edge (9);
\path (l7) edge (10);
\path (l6) edge (12);
\path (l6) edge (11);

\path (l4) edge (l1);
\path (l3) edge (l7);
\path (l3) edge (l1);
\path(l3) edge (l6);
\path (l2) edge (l5);
\path (l1) edge (l7);

\path (l8) edge (11);
\path (l8) edge (l6);
\path (l8) edge (l3);
\path (l8) edge (l1);
\path (l1) edge (l2);

\end{tikzpicture}\caption{{\bf (left)} Instance of $\CacAP$, where dashed edges denote links. The projections of $\ell_1$ are $\{v_7,v_1\}$, $\{v_1,v_8\}$ and $\{v_8,v_9\}$. Link $\ell_2$ is crossing with $\ell_1$ and $\ell_5$. {\bf (right)} The Corresponding Steiner tree instance, where square nodes denote terminals.}
\label{fig:reduction}
\end{figure}

\begin{figure}[H]    
\centering\begin{tikzpicture}
	\tikzstyle{every state}=[fill=white,draw,minimum size=15pt]

\node[state, circle] (l4) {$\ell_4$};
\node[state, rectangle, below of=l4, node distance=1.3cm] (v2) {$v_2$};
\node[state, rectangle, right of=v2, node distance=1.3cm] (v3) {$v_3$};
\node[state, circle, left of=v2, node distance=1.3cm] (l1) {$\ell_1$};
\node[state, circle, below of=l1, node distance=1.3cm, left of=l1, node distance=1.5cm] (l2) {$\ell_2$};
\node[state, circle, below of=l1, node distance=1.3cm, right of=l1, node distance=1.5cm] (l3) {$\ell_3$};

\node[state, rectangle, below of=l2, node distance=1.3cm] (v7) {$v_7$};
\node[state, rectangle, left of=v7, node distance=1.1cm] (v5) {$v_5$};
\node[state, circle, right of=v7, node distance=1.1cm, fill=gray] (l5) {$\ell_5$};

\node[state, rectangle, below of=l5, node distance=1.3cm] (v6) {$v_6$};
\node[state, rectangle, left of=v6, node distance=1.1cm] (v4) {$v_4$};

\node[state, circle, below of=l3, node distance=1.3cm, fill=gray] (l6) {$\ell_6$};
\node[state, circle, right of=l6, node distance=1.3cm, fill=gray] (l7) {$\ell_7$};

\node[state, rectangle, below of=l6, node distance=1.3cm] (v12) {$v_{12}$};
\node[state, rectangle, left of=v12, node distance=1cm] (v11) {$v_{11}$};

\node[state, rectangle, below of=l7, node distance=1.3cm] (v9) {$v_{9}$};
\node[state, rectangle, right of=v9, node distance=1cm] (v10) {$v_{10}$};
\path (l4) edge (l1);
\path (l4) edge (v2);
\path[ultra thick] (l4) edge (v3);
\path (l1) edge (l2);
\path[ultra thick] (l1) edge (l3);
\path[ultra thick] (l2) edge (v5);
\path (l2) edge (v7);
\path (l2) edge (l5);
\path (l3) edge (l6);
\path[ultra thick] (l3) edge (l7);
\path[ultra thick] (l5) edge (v4);
\path (l5) edge (v6);
\path (l6) edge (v11);
\path[ultra thick] (l6) edge (v12);
\path[ultra thick] (l7) edge (v9);
\path (l7) edge (v10);

\end{tikzpicture}\caption{A feasible Steiner tree for the instance of Figure \ref{fig:reduction}, which happens to be well-structured. Bold edges denote a possible marking. One has $m(\ell_3)=e:=\{\ell_3,\ell_7\}$, and $W(e)$ contains $\{v_9,v_{12}\}$, $\{v_9,v_{5}\}$ and $\{v_9,v_{3}\}$. Notice that $w(e)=|W(e)|=d(\ell_3)+d(\ell_1)-1$. Leaf-Steiner nodes are drawn in grey. Here $\ell_2$ (resp., $\ell_3$) is a good (resp., bad) father. Consequently $\ell_5$ (resp., $\ell_6$) is good (resp., bad). A feasible grouping is $g(\ell_2)=\{\ell_2\}$, $g(\ell_3)=\{\ell_3,\ell_7\}$, $g(\ell_1)=\{\ell_1,\ell_6\}$, $g(\ell_4)=\{\ell_4\}$, and $g(\ell_5)=\{\ell_5\}$.}
\label{fig:optimalSteiner}
\end{figure}
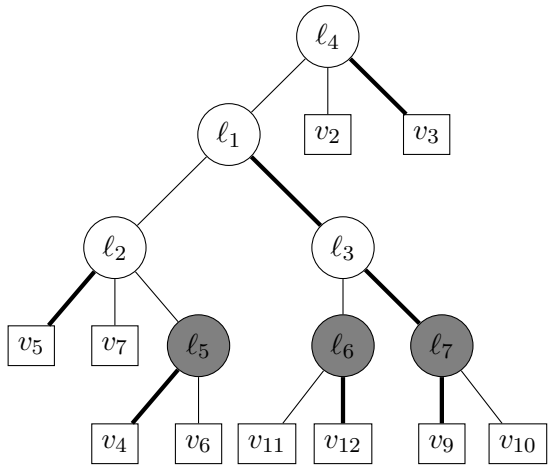

\paragraph{Acknowledgments.} This work is highly in debt with Saket Saurabh. During a visit of the second author to Bergen University a few years ago, Saket mentioned the possibility to use the reduction to Steiner tree to approximate connectivity augmentation problems, possibly with an ad-hoc analysis. The result in this paper follows precisely that path (though finding a good enough analysis was not easy). The second author is also grateful to M. S. Ramanujan and L. Vegh for several helpful discussions on this topic.


\bibliographystyle{abbrv}
\bibliography{references}

\appendix

\section{Omitted Proofs from Section \ref{sec:apx}}

\begin{claim}\label{claim:equality}
$ H_{d_1}+\sum_{j= d_1+1}^{\infty}\frac{1}{j\cdot 2^{j-d_1}}=\hat{H}_{d_1}.$\end{claim}
\begin{proof}
Note that 
$$
\hat{H}_{d_1}=
\sum_{i=0}^{\infty}\frac{H_{d_1+i}}{2^{i+1}}
=\sum_{i=0}^{\infty}\frac{H_{d_1}+\sum_{j=1}^{i}\frac{1}{d_1+j}}{2^{i+1}}=H_{d_1}+\sum_{j=1}^{\infty}\frac{\sum_{i=j}^{\infty}\frac{1}{2^{i+1}}}{d_1+j}
=H_{d_1}+\sum_{j=1}^{\infty}\frac{1}{(d_1+j)2^j}.
$$
\end{proof}
\begin{proof}[Proof of Lemma \ref{lem:badBound}]
The claim is trivially true if $\ell$ is the root since in that case $c(\ell)=H_{d(\ell)-1}<\hH_{d(\ell)}$. So assume $\ell$ is not the root.
For a generic sequence $S=(d_1,\ldots,d_k)$ of positive integers, let us define
$$
f(S)=\sum_{j=1}^{k-1}\frac{(d_{j+1}-1)\cdot H_{d_1+d_2+...+d_j-j+1}}{d_2\cdot d_3 \dots d_{j+1}}+\frac{H_{d_1+d_2+\dots+d_k-k+1}}{d_2\cdot d_3\dots d_k}.
$$
Intuitively, this is the right-hand side of the equation in Lemma \ref{lem:costPerNode}. 
For an infinite sequence $S'=(d_1,d_2,\ldots)$ of positive integers, we analogously define 
$$f(S')=\sum_{j=1}^{\infty}\frac{(d_{j+1}-1)\cdot H_{d_1+d_2+...+d_j-j+1}}{d_2\cdot d_3 \dots d_{j+1}}
$$
Given a finite sequence $S=(d_1,\ldots,d_k)$ of the above type, let $\bar{S}=(d_1,\ldots,d_k,2,2,\ldots)$ be its infinite extension where we add an infinite sequence of $2$ at the end.
\begin{claim}\label{finiteToInfinie}
$f(S)\leq f(\bar{S})$.
\end{claim}  
\begin{proof}
By definition
\begin{eqnarray*}
f(\bar{S})-f(S) & = &  \sum_{j=k}^{\infty}\frac{(d_{j+1}-1)\cdot H_{d_1+d_2+...+d_j-j+1}}{d_2\cdot d_3 \dots d_{j+1}}
-\frac{H_{d_1+d_2+...+d_{k}-k+1}}{d_2\cdot d_3\dots d_{k}}\\
& \geq & \sum_{j=k}^{\infty}\frac{(d_{j+1}-1)\cdot H_{d_1+d_2+...+d_{k}-k+1}}{d_2\cdot d_3 \dots d_{j+1}}
-\frac{H_{d_1+d_2+...+d_{k}-k+1}}{d_2\cdot d_3\dots d_{k}}\\
& =& \frac{H_{d_1+d_2+...+d_{k}-k+1}}{d_2\cdot d_3 \dots d_{k}}\sum_{j=1}^{\infty}\frac{1}{2^j}
-\frac{H_{d_1+d_2+...+d_{k}-k+1}}{d_2\cdot d_3\dots d_{k}}\\
& = & \frac{H_{d_1+d_2+...+d_{k}-k+1}}{d_2\cdot d_3 \dots d_{k}}
-\frac{H_{d_1+d_2+...+d_{k}-k+1}}{d_2\cdot d_3\dots d_{k}}=0.
\end{eqnarray*}
%
\end{proof}
By the above claim it is sufficient to consider infine sequences of type $\bar{S}$. 
We can also assume w.l.o.g. that all $d_i$, $i\geq 2$, in such sequences are at least $2$ by the following claim. 
\begin{claim}\label{lem:infiniteNoOne}
Let $\bar{S}=(d_1,\ldots,d_k,2,2,\ldots)$ and assume there exists $d_i=1$ in the sequence for some $i\geq 2$. Let $\bar{S}_i=(d_1,\ldots,d_i-1,d_i+1,\ldots,d_k,2,2,\ldots)$ be the subsequence where the $i$-th entry is removed. Then $f(\bar{S})=f(\bar{S}_i)$.
\end{claim}
\begin{proof}
Consider the entries in the sum defining $f(\bar{S})$ and $f(\bar{S}_i)$. The entry $j=i-1$ in $f(\bar{S})$ has value $0$. For $j<i-1$, the $j$-th entries in $f(\bar{S})$ and $f(\bar{S}_i)$ are identical. For $j>i-1$, the $j$-th entry in $f(\bar{S})$ is equal to the $j-1$-th entry in $f(\bar{S}_i)$. 
\end{proof}
By the above claims we can focus on infinite sequences $S=(d_1,\ldots,d_k,2,2,\ldots)$ where $d_i\geq 2$ for $i\geq 2$. Let us prove by induction on $k\geq 2$ that $f(S)\leq \hH_{d_1}$. The claim is true by definition for $k=2$. Next consider any $k>2$ and assume the claim is true for all values up to $k-1$. Define $S'=(d_1+d_2-1,d_3,\ldots,d_k,2,2,\ldots)$. By definition and inductive hypothesis:
\begin{eqnarray*}
f(S) & = &H_{d_1}\frac{d_2-1}{d_2}+\frac{f(S')}{d_2}\leq H_{d_1}\frac{d_2-1}{d_2}+\frac{\hat{H}_{d_1+d_2-1}}{d_2}.
\end{eqnarray*}
By Claim \ref{claim:equality},
\begin{eqnarray*}
H_{d_1}\frac{d_2-1}{d_2}+\frac{\hat{H}_{d_1+d_2-1}}{d_2} & = & H_{d_1}\frac{d_2-1}{d_2}+\frac{1}{d_2}\left(H_{d_1+d_2-1}+\sum_{j\geq d_1+d_2}\frac{1}{j\cdot2^{j-d_1-d_2+1}}\right)\\
& = & H_{d_1}+\sum_{j=d_1+1}^{d_1+d_2-1}\frac{1}{j\cdot d_2}+\sum_{j\geq d_1+d_2}\frac{1}{j\cdot d_2\cdot2^{j-d_1-d_2+1}}\\
 & = & H_{d_1}+\sum_{j\geq d_1+1}\frac{\alpha_j}{j},
\end{eqnarray*}
where 
$$
\alpha_j:=\begin{cases}
\frac{1}{d_2} & \text{for } d_1+1\leq j\leq d_1+d_2-1;\\
\frac{1}{j\cdot 2^{i-d_1-d_2+1}} & \text{for } j\geq d_1+d_2.
\end{cases}
$$
We observe the following simple facts about the coefficients $\alpha_j$.
\begin{claim}\label{claim:coefficients}
One has:
\begin{enumerate}
\item $\sum_{j\geq d_1+1}\alpha_j=1$.
\item For every $i>1$, $\sum_{j\geq d_1+i}\alpha_j\geq \frac{1}{2^{i-1}}$.
\end{enumerate}
\end{claim}
\begin{proof}
\begin{enumerate}
     \item $\sum_{j\geq d_1+1}\alpha_j=\frac{d_2-1}{d_2}+\sum_{j=d_1+d_2}^{\infty}\frac{1}{d_2\cdot2^{j-d_1-d_2+1}}=1-\frac{1}{d_2}+\frac{1}{d_2}$.
\item For $i\geq d_2$, one has
$$
\sum_{j\geq d_1+i}\alpha_j=\sum_{j=d_1+i}^{\infty}\frac{1}{d_2\cdot2^{j-d_1+d_2-1}}=\frac{1}{d_2\cdot2^{i-d_2}} \geq \frac{1}{2^{i-1}},
$$     
where in the inequality we used the fact that $k\le2^{k-1}$ for any integer $k\geq 1$.  

For $2\leq i \leq d_2-1$, one has: 
$$
\sum_{j\geq d_1+i}\alpha_j=\frac{d_2-i}{d_2}+\frac{1}{d_2}=\frac{d_2-i+1}{d_2}\geq \frac{1}{i}\geq \frac{1}{2^{i-1}},
$$
where in the first inequality above we used the fact that $\frac{k-j+1}{k}$ is a decreasing function of $k\geq j+1$ and $d_2\geq i+1$, and in the second inequality again the fact that $k\le2^{k-1}$ for $k\geq 1$.
\end{enumerate}
\end{proof}
Intuitively, the term $A=\sum_{j= d_1+1}^{\infty}\frac{\alpha_j}{j}$ is a convex combination of terms of type $1/j$ under the constraint that the sum of the tail coefficients is large enough. An obvious upper bound on $A$ is obtained by choosing coefficients $\beta_j$ that respect the constraints on $\alpha_j$ given by Claim \ref{claim:coefficients}, and at the same time are as large as possible on the smallest terms of the sum. An easy induction shows that the best choice is $\beta_j=\frac{1}{2^{j-d_1}}$ for all $j\geq d_1+1$. Thus we can conclude
$$
f(S)\leq H_{d_1}+\sum_{j\geq d_1+1}\frac{\alpha_j}{j}\leq H_{d_1}+\sum_{j\geq d_1+1}\frac{\beta_j}{j} = H_{d_1}+\sum_{j= d_1+1}^{\infty}\frac{1}{j\cdot 2^{j-d_1}} =\hat{H}_{d_1},
$$
where last equality comes from Claim \ref{claim:equality}.
\end{proof}

\begin{proof}[Proof of Claim \ref{claim:maximum}]
Consider $a_1(i)$. Excluding a fixed additive term $\hat{H}_2-p$, the value of this function is $a'_1(i):=\frac{H_{i+2}-x}{i}$, where $x=\hat{H}_2-p\in (0,\hat{H}_2]$. Taking the discrete derivative
$$
a'_1(i+1)-a'_1(i)=\frac{x+\frac{i-1}{i+3}-H_{i+3}}{i(i+1)}
$$
one might observe that this is negative for $i\geq 6$ since $x+\frac{i-1}{i+3}\leq \hat{H}_2+1<2.7726<H_9> 2.8289$. The reader might skip the following cases that are analogous.

Consider now $a_2(i)$. Excluding a fixed additive term $\hat{H}_2$, the value of this function is $a'_2(i):=\frac{H_{i}-x}{i}$, where $x=\hat{H}_2-p\in (0,\hat{H}_2]$. One has 
$$
a'_2(i+1)-a'_2(i)= \frac{x+1-H_{i+1}}{i(i+1)},
$$
which is negative for $i\geq 8$ since $x+1\leq \hat{H}_2+1<2.7726<H_9> 2.8289$.

Consider next $a_3(i)$. Excluding a fixed additive term $\hat{H}_2-p$, the value of this function is $a'_3(i):=\frac{\hat{H}_{i+2}-\hat{H}_2}{i}$. One has
$$
a'_3(i+1)-a_3'(i)=\frac{\hat{H}_2-\hat{H}_{i+2}}{i(i+1)}+\sum_{j\geq 1}\frac{1}{2^j(i+1)(i+j+2)}\leq \frac{\hat{H}_2+1-\hat{H}_{i+2}}{i(i+1)},
$$
which is negative for $i\geq 6$ since $\hat{H}_2+1<2.7726<\hat{H}_8> 2.8194$.

It remains to consider $a_4(i)$. Excluding a fixed additive term $\hat{H}_2$, the value of this function is $a'_4(i):=\frac{H_{i}-\hat{H}_2}{i}$. One has
$$
a'_4(i+1)-a'_4(i)=\frac{\hat{H}_2-\hat{H}_{i}}{i(i+1)}+\sum_{j\geq 1}\frac{1}{2^j(i+1)(i+j)}\leq \frac{\hat{H}_2+1-\hat{H}_{i}}{i(i+1)},
$$
which is negative for $i\geq 8$ since $\hat{H}_2+1<2.7726<\hat{H}_8> 2.8194$.
\end{proof}

\section{Details on the Reduction to Steiner Tree}
\label{sec:lemma1proof}

\begin{proof}[Proof of Lemma \ref{lem:SteinerReduction}]
$\Leftarrow$ Assume by contradiction that $A$ is not a feasible $\CacAP$ solution. Then there exists a $2$-edge cut $\{e_1,e_2\}$, for two edges $e_1,e_2$ belonging to some cycle $C$ of $G$, which is not covered by any link in $A$. Let $G_L=(V_L,E_L)$ and $G_R=(V_R,E_R)$ be the two (vertex disjoint) connected components identified by this cut. Let also $t_L$ and $t_R$ be any two degree $2$ nodes in $V_L$ and $V_R$, respectively. (Observe that these nodes must exist.) By assumption there exists a (simple) path $P=t_L,\ell_1,\ldots,\ell_q,t_R$ between $t_L$ and $t_R$ in $G_{ST}[T\cup A]$, where all $\ell_i$'s are link nodes. Since $\{e_1,e_2\}$ is not covered, each such link has both endpoints either in $V_L$ or in $V_R$. Furthermore, $\ell_1$ and $\ell_q$ have one endpoint in $V_L$ and $V_R$, resp. Hence there must be two consecutive links $\ell_i$ and $\ell_{i+1}$ where $\ell_i$ has both endpoints in $V_L$ and $\ell_{i+1}$ both endpoints in $V_R$. These links cannot be crossing, therefore contradicting the fact that $\{\ell_i,\ell_{i+1}\}$ is an edge of $G_{ST}$.

$\Rightarrow$ We first observe that, w.l.o.g., we can replace each link $\ell$ with its projections $proj(\ell)$. The feasibility of $A$ is preserved. The same holds for the connected components of $G_{ST}[T\cup A]$ since the links in $proj(\ell)$ induce a path in $G_{ST}$. Thus for simplicity we assume that all links in $A$ have both their endpoints in the same cycle. Let $C_1,\dots,C_k$ be the cycles of $G$. For any cycle $C_i$ of the cactus $G$ let $A_i$ be the set of links in $A$ with both their endpoints in $C_i$. The following lemma shows that $G_{ST}[A_i]$ is connected.
\begin{lemma}\label{lem:ReductionOneCycle}
 Let $G=(V,E)$ be an input cactus of CacAP which consists of exactly one cycle and let $A$ be a feasible solution for $G$.  Then $G_{ST}[A]$ is connected.
\end{lemma}
\begin{proof}
Assume that $G_{ST}[A]$ is not connected. Then $A$ can be partitioned in $L_R$ and $L_B$, such that for any $l_{R}\in L_R$ and $l_{B}\in L_B$, $l_{R}$ does not cross $l_{B}$. We call the links in $L_R$ red links and the links in $L_{B}$ blue links. We can also partition $V$ in $V_R$ and $V_B$, such that the endpoints of red links belong to $V_R$ and the endpoints of blue links belongs to $V_B$. Therefore we call $V_B$ and $V_R$, blue vertices and red vertices respectively.

Let $V_1,V_2,\dots,V_{2k}$ be the partition of vertices of the cycle $G$ into maximal consecutive blocks of vertices of the same color, so that $V_1\cup V_3 \cup \dots \cup V_{2k-1}=V_R$ and $V_2\cup V_4 \cup \dots \cup V_{2k}=V_B$. 

We say that a link $\ell=\{u,w\}\in A$ is \textit{nice}, if $u$ and $v$ belong to different blocks $V_i$ and $V_j$, $i\neq j$. We say that an edge $e=\{u,v\} \in E$ is \textit{colorful} if $u$ is red and $v$ is blue or vice versa. Note that $G$ has precisely $2k$ colorful edges. If there is no nice link in $A$, then any pair of colorful edges of $G$ is not covered by $A$, which is a contradiction.

Assume that $\ell=\{u,v\}\in A$ is a nice link, such that the distance between $u$ and $v$ in the cycle $G$ is minimum. Assume that $u\in V_1$ and $v \in V_{2x+1}$ (and therefore these are red vertices) and also that the vertices of $V_2$ are in the shortest path from $u$ to $v$. Now let $e_1$ and $e_2$ be the colorful edges such that exactly one of their endpoints is in $V_2$. Now we show that the cut formed by $e_1$ and $e_2$ is not covered by $A$. 

Assume that $\{e_1,e_2\}$ is covered, then there should be a link $\ell_1=(w,z)$ such that $w\in V_2$ and $z\not\in V_2$. Then either this link is a nice link that crosses $\ell$, which is a contradiction since $\ell\in L_R$ and $\ell_1\in L_B$, or $\ell_1$ is a nice link such that the distance of $u$ and $v$ is less than the distance of $w$ and $z$.

\end{proof}
For every pair of cycles $C_i$ and $C_j$ that share a vertex $w$, there is a link $\ell_i\in A_i$ and $\ell_j\in A_j$ which are incident to $v$, thus $\ell_{i}$ and $\ell_{j}$ cross. We can conclude that $G_{ST}[A]$ is connected. Finally, since $A$ is feasible, there exists at least one link $\ell\in A$ incident to each node $t$ of degree $2$ in $G$, which implies that the edge $\{\ell,t\}$ belongs to $E_{ST}$. Thus $G_{ST}[T\cup A]$ is also connected.

\end{proof}

\section{Some Details About the Steiner Tree Approximation Algorithm in \cite{BGRS13}}
\label{sec:backgroundSteiner}

We will briefly discuss the $\ln(4)+\epsilon$ approximation algorithm from \cite{BGRS13} for the Steiner tree problem. For a complete presentation of the Steiner tree algorithm we refer to the original paper \cite{BGRS13}. The algorithm is based on the Directed Component Relaxation (DCR) of the Steiner tree problem.

\begin{align}
\text{min} \quad \sum_{C \in \calC}  c(C)x_{C} & \quad \quad \text{(DCR)} \\
 \text{s.t.} \quad \sum_{C \in \delta^+_{\calC}(U)}  x_{C} & \geq 1  \quad\quad \forall  \emptyset \neq U \subseteq T \setminus \{ r \}\\  
 x_{C} & \geq  0 \quad\quad \forall ~C \in \calC.
\end{align}


Here $\calC$ is a set of directed components, where each directed component $C$ is a minimum-cost Steiner tree (of cost $c(C)$) over a subset of terminals. Furthermore, the leaves of $C$ are precisely its terminals, and $C$ is directed towards a specific terminal: the latter node is the sink of $C$, and the remaining terminals are the sources of $C$. Intuitively, our goal is to buy a minimum-cost subset of directed components so that they induce a directed path from each terminal to the root. In more detail, for any cut $U$ that separates some non-root terminal from the root, let $\delta^+_{\calC}(U)$ be the set of components with some source in $U$ and the sink not in $U$. Then every feasible solution has to buy some component in $\delta^+_{\calC}(U)$. The DCR relaxation follows naturally.

%
%

After restricting DCR to solutions that only use components with at most $k$ terminals we obtain DCR$_k$.
For constant $k$, DCR$_k$ has a polynomial number of variables. Furthermore, the separation problem can be solved in polynomial time via a reduction to minimum cut. Therefore DCR$_k$ can be solved in polynomial time. Moreover, the value of DCR$_k$ is known to be a $(1+\epsilon)$-approximation of the value of DCR for large enough $k=O_\eps(1)$.

The iterative randomised rounding algorithm from \cite{BGRS13}, until all terminals are connected to the root, in iterations $t=1,2,3 \ldots$, does the following:
\begin{itemize}
 \item solve DCR$_k$ for the current instance of the Steiner tree problem to get $x^t$;
 \item sample a component $C^t$ from $\calC_k$ with probability proportional to $x_C^t$;
 \item contract the sampled component $C^t$.
\end{itemize}

For the ease of the analysis, by adding dummy components w.l.o.g, one may assume that the total number of components in the fractional solution remains constant across the iterations of the algorithm, i.e., $\sum_{C \in \calC}x^t_{C} = M$ for a proper $M$ for all $t=1,2,\ldots$.
It is argued that after $t$ iterations of the algorithm, having bought the first $t$ sampled components, the residual instance of the problem is expected to be less costly. To this end a reference solution $S^t$ is constructed such that $S^t \cup \bigcup_{t'=1}^{t-1} C^t$ connects all the terminals. The initial reference solution $S^1=OPT_{ST}$ is an optimal solution to the Steiner tree instance of cost $opt$. Consecutive reference solutions $S^{2}, S^{3}, \ldots $ are obtained by gradually deleting edges that are no longer necessary due to the connectivity provided by the already sampled components. 

Key to estimate the expected cost of the final solution is to bound the number of iterations until a particular edge $e \in S^1$ can be removed. Define $D(e) = \max\{t | e \in S^t\}$. In \cite{BGRS13} (proof of Theorem 21) it is shown that there exist a randomised process of constructing reference solutions $S^1, S^2, \ldots$ such that $E[D(e)] \leq \ln(4) \cdot M$, which allows one to bound the total expected cost of sampled components as $E \left[ \sum_{t\geq1}^{} c(C^t) \right] \leq (ln(4)+\epsilon) \cdot opt$. Note that the above \emph{per-edge} guaranty allows for easily handling arbitrary costs of individual edges. In our application to (unweighted) $\CacAP$, we need to average over multiple edges to achieve a good enough bound.

\subsection{Witness Tree and Witness Sets}

We next slightly abuse notation and sometimes denote in the same way a tree and its set of edges. The construction of reference solutions $S^1, S^2, \ldots$ is not trivial. It involves:
\begin{itemize}\itemsep0pt
 \item construction of a terminal spanning tree $W$, called the \emph{witness tree}, based on randomised marking (selection) of a subset of edges of $S^1$. Each edge $e$ of $S^1$ is associated with a proper subset $W(e)\subseteq W$ (witness set of $e$);
 \item randomised deletion of a proper subset of $W$ in response to selecting a particular component $C^t$ in iteration $t$;
 \item removing an edge $e$ from $S^t$ when all edges $W(e)$ have already been deleted. 
\end{itemize}

In the following we discuss the main idea behind and the key properties of each of the three above mentioned processes. We also pin-point the element of the analysis that can be modified in order to utilise the specific properties of the instance we obtain from the reduction from $\CacAP$.

\paragraph{Construction of the witness tree.}
The high level idea behind the witness tree is that we need to always satisfy the condition that $S^t \cup \bigcup_{t'=1}^{t-1} C^t$ connects all the terminals, which is that the remaining fragments of the initial reference solution $S^1$ together with the already sampled components must provide sufficient connectivity. To this end a simpler object providing connectivity is constructed. It is an auxiliary tree $W$ whose node set is the terminals of the instance (edges of  $W$ are independent of the edges of the input graph). It will be easier to delete edges from $W$ in response to sampling components rather than deleting them directly form $S^t$.

We will now discuss methods to construct $W$. Intuitively, removing edges from a Steiner tree (in response to receiving connectivity from a component) is directly possible for only a subset of edges of the Steiner tree.
In particular it appears more difficult to remove a Steiner vertex (and hence a path connecting a Steiner vertex to a terminal). This is related to the concept of \emph{Loss} and \emph{Loss contracting algorithms} (see, e.g., \cite{robins2005tighter}), where one accepts that the cost of the system of paths connecting Steiner nodes to terminals is not removable.

Consider the following procedure: For each component\footnote{Recall that a full component is a maximal subtree whose terminals are exactly its leaves.} $S'$ of the Steiner tree $S^1$ select a single Steiner vertex $s$ and draw the component as a tree rooted in $s$. For every Steiner vertex of $S'$ select and mark a single edge going down (away from $s$). Note that for each Steiner vertex $v$ the marked edges will form a unique path towards a leaf containing terminal $t(v)$. Note also that connected components formed by the marked edges will all have a single terminal node. Construct $W(S')$ by adding to $E(W(S'))$ an edge $\{ t(u), t(v) \}$ for each unmarked edge $\{u, v \}$ of $S'$.\footnote{Note that in \cite{BGRS13} the role of marked and unmarked edges was reversed. It was irrelevant for the analysis in \cite{BGRS13} as it was assumed that the tree $S'$ is binary. In this paper however we will exploit the high degree of Steiner nodes in $S'$ and hence prefer to mark the "Loss" edges.} Observe that the above constructed graph $W(S')$ is a tree spanning the terminals of $S'$. By repeating this procedure for all full components of $S^1$ we obtain tree $W$ spanning all terminals of the Steiner tree instance.

So far we did not specify how to select the edge below Steiner node $v \in S'$ to be marked. In \cite{BGRS13} the tree was assumed to be binary, and the edge would be selected at random by tossing a fair coin. In the current paper we use a different marking strategy as discussed in Section \ref{sec:marking}. 

\paragraph{Marking edges of the witness tree.}

When edges of the witness tree $W$ become unnecessary, we mark them. We keep the invariant that the unmarked edges of $W$ together with the already collected components are sufficient to connect all terminals. Still, given a fixed collection of the already sampled components, the choice of which edges of $W$ to mark is not obvious. In \cite{BGRS13} a randomised marking scheme was considered. It was shown (Lemma 19 in \cite{BGRS13}) that there exists a random process marking edges in $W$ in response to sampled components, such that for every edge $e \in W$ not marked until iteration $t$ the probability that it is marked in iteration $t$ is at least $1/M$. In the current work we continue using the mentioned ``uniform'' witness tree marking process, and utilise the following lemma.

\begin{lemma}[lemma 20 in \cite{BGRS13}]
 Let $\tilde{W} \subseteq W$. Then the expected number of iterations until all edges in $\tilde{W}$ are marked is at most $H_{|\tilde{W}|} \cdot M$.
\end{lemma}

\paragraph{Removing edges of the reference tree $S^t$.}

Which edges of the reference tree can be removed? Clearly it suffices if $S^t$ provides the same terminal connectivity as the unmarked edges of the witness tree $W$. Note that a single edge $e \in W$ corresponds to a single path $p(e)$ in $S^1$. It then suffices to keep the edges of $S^1$ that occur in a path $p(e)$ of at least one unmarked edge $e \in W$.

We introduce the following notation: for an edge $f$ in $S^1$ let $W(f) = \{ e \in W | f \in p(e) \} $, we call $W(f)$ to be the \emph{witness set} of $f$. Therefore, at iteration $t$, the reference solution $S^t$ contains the edges form $S^1$ whose witness sets are not fully marked until iteration $t-1$.

Observe that the expected number of iterations an edge $f$ from the reference solution survives (until being removed) $E[D(f)]$ can be expressed using only the size of its witness set $W(f)$.

\begin{corollary}
 Let $f \in S^1$, then $E[D(f)] \leq H_{|W(f)|} \cdot M$.
\end{corollary}

Following the argument from the proof of Theorem 21 in \cite{BGRS13}, we also get

\begin{corollary}
For $k=O_\eps(1)$ large enough, the total cost of components bought by the algorithm is at most
 \[
 \frac{1 + \eps}{M} \sum_{f \in S^1}^{} E[D(f)] \cdot c(f) \leq (1+\epsilon) \cdot \sum_{f \in S^1}^{} H_{|W(f)|} \cdot c(f)
 \]
\end{corollary}

Therefore, it suffices to analyse how the marking scheme used in the construction of the witness tree affects distributions of the sizes of the witness sets for the individual edges of $S^1$.  To this end we will exploit two properties of our instances: the high degree of the Steiner vertices in the optimal solutions, and the fact that all edges of $S^1$ have the same cost.

\end{document}